\documentclass{amsart}

\usepackage[utf8]{inputenc}
\usepackage{amsmath, amssymb, amsbsy, amsthm, amsaddr}
\usepackage{graphicx}

\newcommand{\N}{\mathbb{N}}
\newcommand{\R}{\mathbb{R}}

\newcommand{\eps}{\varepsilon}

\renewcommand{\phi}{\varphi}

\newcommand{\calO}{\mathcal{O}}

\newcommand{\calR}{\mathcal{R}}

\renewcommand{\bar}[1]{\overline{#1}}

\newcommand{\norm}[1]{\left\|#1\right\|}

\DeclareMathOperator{\pder}{\partial}

\newcommand{\lb}{\left(}
\newcommand{\lsb}{\left[}
\newcommand{\lcb}{\left\{}

\newcommand{\rb}{\right)}
\newcommand{\rsb}{\right]}
\newcommand{\rcb}{\right\}}

\newcommand{\lpm}{\begin{pmatrix}}
\newcommand{\rpm}{\end{pmatrix}}

\newtheorem{theorem}{Theorem}
\newtheorem{lemma}[theorem]{Lemma}

\begin{document}

\title[Age--dependent dengue model: Analysis and comparison to field data] {An Age--Dependent Model for Dengue Transmission: Analysis and Comparison to Field Data from Semarang, Indonesia}
\author{N.C.~Ganegoda}
\address{Department of Mathematics, University of Sri Jayewardenepura, Sri Lanka}
\email{naleencg@gmail.com}

\author{T.~Götz and K.P.~Wijaya}
\address{Mathematical Institute, University of Koblenz--Landau, 56070 Koblenz, Germany}
\email[]{goetz@uni-koblenz.de (Corresponding author)}
\email{karuniaputra@uni-koblenz.de}



\begin{abstract}
Medical statistics reveal a significant dependence of hospitalized dengue patient on the patient's age. To incorporate an age--dependence into a mathematical model, we extend the classical ODE system of disease dynamics to a PDE system. The equilibrium distribution is then determined by the fixed points of resulting integro--differential equations. In this paper we use an extension of the concept of the basic reproductive number to characterize parameter regimes, where either only the disease--free or an  endemic equilibrium exists. Using rather general and minimal assumptions on the population distribution and on the age--dependent transmission rate, we prove the existence of those equilibria. Furthermore, we are able to prove the convergence of an iteration scheme to compute the endemic equilibrium. To validate our model, we use existing data from the city of Semarang, Indonesia for comparison and to identify the model parameters.

\smallskip
\textbf{Keywords:}
Age--dependent disease dynamics, equilibria, integro--differential equation, parameter estimation.
\end{abstract}

\maketitle

\section{Introduction}

In the current fashion of semi-linear hyperbolic PDE, age-structured population models have been appearing in the literature likely since 1972, where the pioneering model served to describe human demography~\cite{Coa72}. Due to high practical relevance, ensuing models have thus been directed to understand disease epidemics~\cite{Hop75,DS85,MD86}. For quite recent work, one may refer  to~\cite{LEK14,RWK15,WCL18}. Aspects regarding disease characteristics that are relevant to the modeling have been brought on the table, for example, vertical transmission, mode of disease trasmission from infectives to susceptibles (intercohort or intracohort), disease-related age-dependent deaths, age-dependent recovery rate, age-dependent natural death rate, measure on disease persistence (the \emph{basic reproductive number}), host--vector interaction, and bite structure. As high-throughput data are nowadays becoming more available, not only can one fix the existing models with close forms of parameters involved, but data assimilation practices can also be assigned to produce results that are readable throughout different disciplines. On the other hand, understanding data quality is also necessary to develop preliminary assessment before putting the data into modeling. Many previous papers on age-dependent models counted heavily on model novelties and the corresponding analytical aspects as well as numerical visualizations~\cite{LEK14,RWK15,WCL18,BCI88,Ina90,IMP92,Cap93,IS04,FP08}. Most of them have left some parameters such as the infection rate and natural death rate appear in general forms, equipped with certain cost-to-go criteria for satisfactory analytical results. Notwithstanding data availability, these criteria have sensibly been built to generate broad feasible spaces within which optimal parameters resulted from data assimilation might instantly lie. These typically include positivity, boundedness (be it in the pointwise sense or in the sense of Lebesgue), and piecewise continuity.

Here we rest our novelty on two pillars: parameter modeling and data assimilation. The underlying model has been taken from the usual host-vector model with age structure in host population. Due to data unavailability, we exclude vector population in the model by the aid of time scale separation, in which way one can assume constancy of the population under a sufficiently large time scale. The resulting model is an SI model, which includes additional information in new parameters regarding the presence of vector population. Models to some parameters are thus initiated to let the model solution capture the dynamics of the data, without having to let the parameters timely unobservable. We use the data of annually recorded hospitalized cases in the city of Semarang, Central Java, Indonesia from 2009 to 2014~\cite{DS14}. Due to such relatively short time horizon inherent from the data, we consider it inappropriate to model the whole data, where otherwise gives insensible prediction. We instead bring both the data and model to their steady states. The former was done by averaging the age-structured data in time, while the the latter is a feasible task. A model for the infection rate is thus inspired by the curve of this steady state of the data. A core iterative method for solving the steady state equation is done for a certain set of values to unobservable parameters, where the solution need not assimilate the steady state data. Optimal values for the parameters are thus found with respect to the variation of this set in the iterative manner. Further, an SIR counterpart is tested for a performance comparison, wherein the same procedure for finding optimal parameters is applied.

\section{Mathematical modeling}

\subsection{An age-structured SISUV model}

Let $p(t,x)=s(t,x)+i(t,x)$ denote the total host population at time $t$ and of age $x$, subdivided into the susceptible and infected class. The human mortality rate $\mu(x)$ is obviously dependent on age, however is assumed to be independent of the infection status. The maximal age of hosts is denoted by $A$. The epidemic investigated here assumes no vertical transmission, i.e.~all newborn hosts are susceptible. The corresponding birth rate of the hosts is denoted by $b(x)$. The disease is transmitted from an infected vector to a susceptible host of age $x$ at a certain age-dependent infection rate $\vartheta(x)$, assumed to be continuous. According to \cite{EV98}, such infection rate folds terms such as mosquito biting rate and transmission rate, where the latter represents the probability that a single mosquito bite leads to a successful infection. The biting rate depends on host mobility that exposes them to vector agglomeration sites, while the transmission rate depends on host immunity that has been shown to be dependent on age~\cite{JC10}. We assume that infected hosts of any age lose immunity at a uniform rate $\gamma$. Regarding the vector population, age structure is ignorable due to a shorter lifespan of the vectors as compared to that of the hosts. The time-varying population $M(t)=u(t)+v(t)$ comprises the susceptible and virus-carrying class. For simplicity, we assume that $M$ follows the logistic-like equation $\text{d}M\slash\text{d}t=\lambda-\rho M$, where $\lambda,\rho$ denote the number of newborn vectors per unit time (recruitment rate) and vector mortality rate, respectively. At a rate $\sigma(y)$, a virus is transmitted from an infected host of age $y$ to a susceptible vector. The preceding description leads us to the following coupled PDE--ODE model
\begin{subequations}
\begin{align}
	\pder_t s(t,x) + \pder_x s(t,x) &= -\frac{1}{N} \vartheta(x) s(t,x) v(t) - \mu(x) s(t,x)+\gamma i(t,x),\\
	\pder_t i(t,x) + \pder_x i(t,x) &= \frac{1}{N} \vartheta(x) s(t,x) v(t) - (\mu(x)+\gamma) i(t,x),\\
	\frac{\text{d}}{\text{d}t}u(t) &= \lambda- u(t)  \int_0^{A} \sigma(y) i(t,y)\, \text{d}y 
		-\rho u(t),\\
	\frac{\text{d}}{\text{d}t}v(t) &= u(t)\int_0^{A} \sigma(y) i(t,y)\, \text{d}y-\rho v(t).
\end{align}
\end{subequations}	
equipped by the initial condition
\begin{equation*}
s(t,0) = \int_0^{A} b(y) p(t,y)\, \text{d}y,\quad i(t,0)=0,\quad u(0)=u_0,\quad v(0)=v_0.
\end{equation*}

The scaling factor $\frac{1}{N}$ in the model equations captures the theoretical view of realistic mass action in manipulating transmission of infection according to population size~\cite{JDH95}. The total host population itself satisfies the PDE
\begin{alignat}{2}
	\pder_t p(t,x) + \pder_x p(t,x) &= -\mu(x) p(t,x), 
		\qquad& p(t,0) &= \int_0^{A} b(y) p(t,y)\,\text{d}y. \label{E:PDEp}
\end{alignat}
In the sequel, we consider the host population in the steady state. In particular, we assume a constant total host population $N:=\int_0^{A} p(x)\, \text{d}x$. Ignoring the time variation in \eqref{E:PDEp}, we arrive at the explicit solution
\begin{equation}\label{E:mu}
	p(x) = N_0 \exp \lb -\int_0^x \mu(y)\, \text{d}y \rb
\end{equation}
where $N_0=p(0)=\int_0^{A} b(y) p(y)\, \text{d}y$ denotes the total number of newborns. Relying on its formulation, it is quite reasonable to designate that $p$ is a positive continuous function. 

\subsection{Model reduction}
For the sufficiently large time scale undertaken, we assume that the vector population cannot vary significantly within one unit time. Under a settling situation, the population ideally portrays the inherent nontrivial equilibrium. Due to its appearance in the host dynamics, we particularly are interested in the virus-carrying state
\begin{equation*}
	v^\ast=\frac{\lambda}{\rho} \frac{J}{(J+\rho)}=  \frac{\lambda}{\rho^2}J \lb 1+ \frac{J}{\rho}\rb^{-1}	 \qquad
	\text{where} \quad  J=\int_0^{A} \sigma(y)i(y)\, \text{d}y. 
\end{equation*}
For the reason that hospitalized hosts $i(y)$ are isolated in hygienic rooms, the presence of vectors is considered less likely. The biting rate of susceptible vectors to those hosts is therefore negligible, inducing $\sigma(y)\simeq 0$ or $J/\rho\ll 1$. One can naturally neglect higher order terms of $J/\rho$ in the Taylor expansion of $v^\ast$ and obtain
\begin{equation}
	v^\ast \simeq \frac{\lambda}{\rho^2}J=\frac{\lambda}{\rho^2} \int_0^{A} \sigma(y) i(y)\, \text{d}y
	= \frac{\lambda\, \tilde{\sigma}}{\rho^2} \int_0^{A} i(y)\, \text{d}y
\end{equation}
for some intermediate $\tilde{\sigma}=\sigma(\tilde{y})$ with $\tilde{y}\in (0,A)$. Using the equilibrium $v^\ast$ together with $s(x)=p(x)-i(x)$ and setting $w(x):=i(x)/p(x)$, we arrive at the equilibrium of the host dynamics 
\begin{subequations}
\label{E:equiSI-PDE}
\begin{alignat}{2}
	w' &= - \gamma w + \theta(x) (1-w) Q(w), \qquad & w(0)&=0, \label{E:ODE:w}\\
	Q(w) &=\frac{1}{N} \int_0^{A} p(y) w(y)\, \text{d}y.
\end{alignat}
\end{subequations}
Here, the new infection rate is defined as $\theta(x) :=\vartheta(x)\lambda\tilde{\sigma}/\rho^2$.

\subsection{Models for the total population and infection rate}

\begin{figure}[htbp!]
\centerline{\includegraphics[width=.67\textwidth]{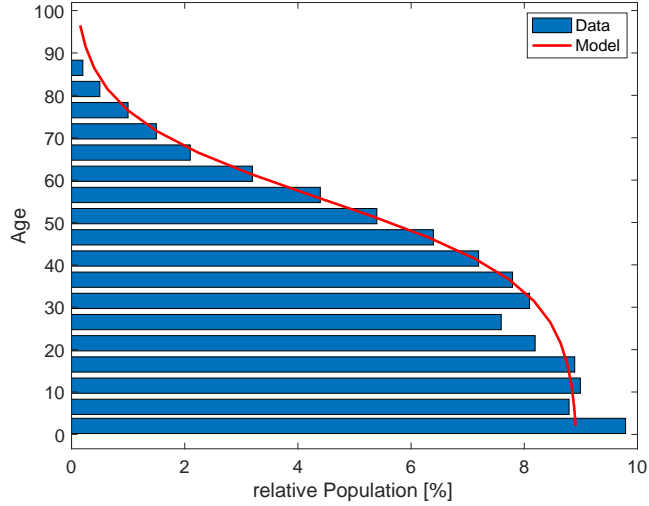}}
\caption{\label{F:agepyr} Data (year 2016; summed over both sexes) for the age pyramid (blue) and fitted parametric model~\eqref{E:pinfty} (red).}
\end{figure}

The data of hospitalized cases are recorded from the city of Semarang, Indonesia. To accompany the data assimilation using these data, we first fit a parametric model
\begin{equation}
	\label{E:pinfty}
	p(x;k_0,k_1,k_2) = \frac{k_0}{1+e^{k_1(x-k_2)}}
\end{equation}
to available population data with age structure in Indonesia, see~\cite{PPI16}. A comparison of the age pyramid in the year 2016 with our model~\eqref{E:pinfty} is shown in Figure~\ref{F:agepyr}. The fitted parameter values are $k_0 = 0.0181\cdot N$, where $N=1.6\cdot 10^6$ equals the total population of Semarang, and $k_1=0.0984$, and $k_2=55.3$. Based on this model, we can set 
\begin{equation}
A=100
\end{equation}
as the maximal age; the population share older than $A$ is thus negligible. The mortality rate $\mu(x)$ in a stationary situation can be obtained from~\eqref{E:mu}:
\begin{equation*}
	\mu(x) = - \frac{\text{d}}{\text{d}x} \ln p(x;k) 
		= \frac{k_1 e^{k_1(x-k_2)}}{1+e^{k_1(x-k_2)}}
		= \frac{k_1}{1+e^{-k_1(x-k_2)}}.
\end{equation*}

For the age-dependent transmission rate $\theta$, we assume that individuals of age around a certain reference value $x_p$ are at higher risk of getting infected than others. We consider a model for $\theta$ as a background transmission rate $\theta_0$ superimposed by a Gaussian peak of height $\theta_1$ centered around $x_p$ and with a width given by the parameter $\sigma$:
\begin{equation}
\label{E:thetaGauss}
	\theta(x;\theta_0, \theta_1, x_p, \sigma) = \theta_0 + \theta_1\cdot e^{-(x-x_p)^2/(2\sigma^2)}.
\end{equation}   
The parameters $\theta_0, \theta_1$ and $\sigma$ are naturally assumed to be positive. The reference value $x_p$ might also be negative indicating that the transmission rate is maximal for new borns and decreases with age. Despite its simple structure, the Gaussian peak model~\eqref{E:thetaGauss} is able to model the observed age distribution of dengue cases in the city of Semarang with quite high accuracy.

In Fig.~\ref{F:agestat} we show the stationary distribution of the infected compartment $i(x)=w(x) p(x)$ (in solid blue) for a given infection rate $\theta(x)$ (in dashed red).

\begin{figure}[htb]
\centerline{\includegraphics[width=0.67\textwidth]{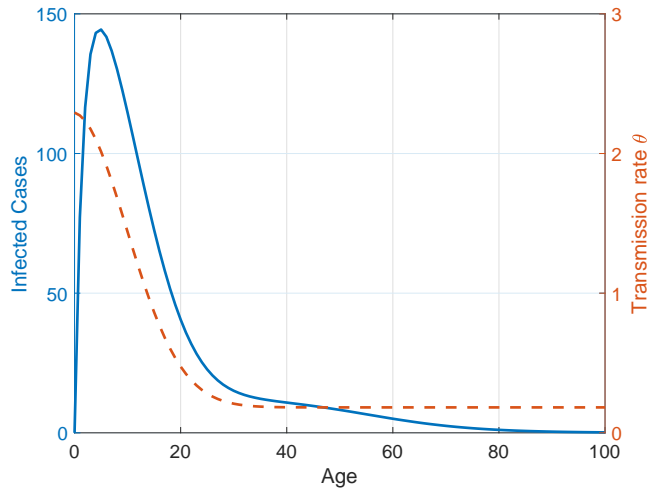}}
\caption{\label{F:agestat} Typical shape of the stationary age distribution of the infected class (solid blue) and the underlying transmission rate $\theta$ (dashed red).}
\end{figure}

\section{Fixed point equation for the equilibrium}

In this section we analyze the solvability of the equilibrium problem~\eqref{E:equiSI-PDE}. 
To prove the existence of equilibria and convergence for an according iteration scheme in a rather generalized setting, we meet the following minimal assumptions on the coefficients and parameters:
\begin{enumerate}
\item Let $\gamma, A>0$; i.e.~the rate of loss of infection and the maximal age are positive.
\item Let $\theta, p:[0,A]\to \R_+$ be continuous and positive functions.
\item Let $\frac{1}{N}\int_0^A p(x)\, dx=1$, i.e.~the population is normalized.
\end{enumerate}

For brevity, we also introduce the following notations
\begin{align*}
	\theta_{\max} &:= \max_{x\in [0,A]} \theta(x),\\
	\Lambda(x,s) &:= \int_s^x \theta(t)\, dt,\\
	k(x,s) &:= \frac{p(x)}{N}\cdot \theta(s)\cdot e^{-\gamma(x-s)} > 0,
\intertext{and the triangle}
	\Omega &= \lcb (x,s)\in [0,A]^2:\, 0\le s\le x \le A\rcb \subset \R^2.
\end{align*}
Note that $\Lambda(x,s) > 0$ for $(x,s)\in \Omega$.

In analogy to classical epidemiological models, we define the \emph{basic reproductive number} as
\begin{equation}
	\calR_0 = \iint_\Omega k(x,s)\, d^2(s,x) >0.
\end{equation}

With these preparations we can now list preliminary lemmata.

\begin{lemma}
For any $Q\in \R$, the solution $w(x;Q)$ of the ODE~\eqref{E:ODE:w} is given by
\begin{equation}
\label{E:w}
	w(x) = w(x;Q) 
		= Q e^{-\gamma x} \int_0^x e^{\gamma s}\, \theta(s) e^{-Q\Lambda(x,s)}\, ds\;.
\end{equation}
The parameter $Q$ itself satisfies the scalar fixed point equation
\begin{equation}
\label{E:FPEQ}
	Q = \frac{1}{N} \int_0^A p(x)\cdot w(x;Q)\, dx 
		= Q \iint_\Omega k(x,s)\cdot e^{-Q\Lambda(x,s)}\, d^2(s,x).
\end{equation}
\end{lemma}

\begin{lemma}
\label{L:w1}
Let $Q>0$. Then $0<w(x;Q)<1$ for all $x>0$ and $w(x;Q)=0$ if and only if $x=0$. Furthermore, $Q>0$ implies $Q<1$.
\end{lemma}

Solutions of the SI model~\eqref{E:equiSI-PDE} or the fixed point problem~\eqref{E:FPEQ}, which are of significant practical relevance require nonnegative populations, i.e. $w,Q\ge 0$. Later on, we also need the following estimate for $Q=1$.
\begin{lemma}
For all $x\in [0,A]$, it holds that $w(x;1) < 1 - \frac{\gamma}{\gamma+\theta_{\max}}$.
\end{lemma}
\begin{proof}
Due to Lemma~\ref{L:w1}, we know that $w$ is bounded by $1$ from above. Suppose that $w(x;1)$ attains a maximum at $x=\xi$, i.e.~$w'(\xi;1)=-(\gamma+\theta(\xi)) w(\xi;1) + \theta(\xi)=0$. Then
\begin{equation*}
	w(x;1)\le w(\xi;1) = \frac{\theta(\xi)}{\gamma+\theta(\xi)} = 1- \frac{\gamma}{\gamma+\theta(x)} < 1 - \frac{\gamma}{\gamma+\theta_{\max}}.
\end{equation*}
\end{proof}

\subsection{Existence of equilibria}

We define the mapping $F:\R\to \R$ by
\begin{equation}
\label{E:DefF}
	F(Q) := 1-\iint_\Omega k(x,s)\cdot e^{-Q\Lambda(x,s)}\, d^2(s,x).
\end{equation}
Solutions to the fixed point problem~\eqref{E:FPEQ} are either given by the trivial solution $Q=0$ or satisfy 
\begin{equation*}
	F(Q) = 0.
\end{equation*}

\begin{lemma}
\label{L:F0:F1}
The function $F$ satisfies $F(0)=1-\calR_0$ and $F(1)>\dfrac{\gamma}{\gamma+\theta_{\max}}>0$.
\end{lemma}
\begin{proof}
The statement for $F(0)$ is obvious and the other is due to 
\begin{align*}
	F(1) &= 1-\frac{1}{N} \int_0^A p(x)\, w(x;1)\, dx \\
		& > 1- \frac{1}{N} \int_0^A p(x) \lb 1-\frac{\gamma}{\gamma+\theta_{\max}}\rb\, dx 
		= \frac{\gamma}{\gamma+\theta_{\max}} \cdot \frac{1}{N} \int_0^A p(x)\, dx \\
		&= \frac{\gamma}{\gamma+\theta_{\max}}.
\end{align*}
\end{proof}

\begin{lemma}
\label{L:dF}
The function $F$ is differentiable and satisfies for all $Q\in (0,1]$:
\begin{equation}
	0\le F'(Q) = \iint_\Omega k(x,s)\cdot \Lambda(x,s)\cdot e^{-Q\Lambda(x,s)}\, d^2(s,x)
	< \calR_0 \frac{e^{-1}}{Q} \simeq \calR_0 \frac{0.37}{Q}.
\end{equation}
\end{lemma}
\begin{proof}
The nonnegativity of $F'$ is obvious. On the other hand, the function $g_1(y;Q):=y\cdot e^{-Qy}$ is bounded with respect to~$y$ by $0\le g_1(y;Q)\le \frac{e^{-1}}{Q}$, where the maximum is attained for $y=1/Q$. Therefore,
\begin{equation*}
	F'(Q) = \iint_\Omega k(x,s)\cdot \Lambda(x,s)\cdot e^{-Q\Lambda(x,s)}\, d^2(s,x)
	\le \frac{e^{-1}}{Q} \iint_\Omega k(x,s)\, d^2(s,x) = \calR_0 \frac{e^{-1}}{Q}\;.
\end{equation*}
\end{proof}

\begin{theorem}[Existence of solutions] For the fixed point problem~\eqref{E:FPEQ}, it holds that
\begin{enumerate}
\item For $0<\calR_0 \le 1$, there exists \emph{only} the \emph{trivial solution} $Q=0$.
\item For $\calR_0>1$, there exists additionally a \emph{unique non-trivial} solution $Q^\ast\in (0,1)$ satisfying $F(Q^\ast)=0$.
\end{enumerate}
\end{theorem}
\begin{proof}
Thanks to Lemma~\ref{L:F0:F1} and the strict monotonicity of $F$ (cf.~Lemma~\ref{L:dF}) it is obvious, that for $0<\calR_0 \le 1$, there exists no solution of $F(Q)=0$ on the relevant interval $Q\in [0,1]$. On the other hand, for $\calR_0>1$, the strict monotonicity of $F$ guarantees the uniqueness of the non-trivial solution $Q^\ast\in (0,1)$.
\end{proof}

Finally, we consider the behavior of the non-trivial solution $Q^\ast$ in the limit $\calR_0 \to 1^+$. Let $\bar{Q} := -\frac{F(0)}{F(1)-F(0)} = \frac{\calR_0-1}{F(1)+\calR_0-1}$ denote the root of the secant connecting $F(0)$ and $F(1)$.

\begin{theorem}
Let $\calR_0>1$. Then 
\begin{equation*}
	0 < Q^\ast < \bar{Q} < \frac{\calR_0-1}{\frac{\gamma}{\gamma+\theta_{\max}}+\calR_0-1}
		\to 0 \quad \text{for $\calR_0\to 1^+$}.
\end{equation*}
\end{theorem}
\begin{proof}
Since $F''(Q) = -\iint_\Omega k(x,s)\cdot \Lambda(x,s)^2\, e^{-Q\Lambda(x,s)}\, d^2(s,x)<0$, we  observe that $F$ is strictly concave and hence $0<Q^\ast<\bar{Q}$. Using the estimate $F(1)>\dfrac{\gamma}{\gamma+\theta_{\max}}$, we yield the desired result owing to Lemma~\ref{L:F0:F1}. 
\end{proof}

As $\calR_0$ get larger than $1$, the trivial fixed point $Q=0$ bifurcates and the non-trivial solution $Q^\ast(\calR_0)<\frac{\calR_0-1}{\frac{\gamma}{\gamma+\theta_{\max}}+\calR_0-1}$ emerges. 


\subsection{Convergence of a fixed point iteration scheme}

To solve the original problem~\eqref{E:equiSI-PDE} we propose an iterative scheme:
\begin{enumerate}
\item Given an initial guess $Q^{(0)} > 0$.
\item For $j=1,\dots$ solve the ODE
\begin{equation*}
	\frac{d}{dx} w^{(j)} 
		= - \gamma w^{(j)} + Q^{(j-1)} \theta(x)\lb 1-w^{(j)}\rb\,, 
		\quad w^{(j)}(0)=0\;.
\end{equation*}
\item Set $\displaystyle Q^{(j)} := \frac{1}{N}\int_0^A p(y) w^{(j)}(y)\,dy$.
\end{enumerate}
The initial guess $Q^{(0)}=0$ can be excluded since in this case we directly obtain $w^{(j)}\equiv 0$ and hence $Q^{(j)}=0$ for all $j\in \N$.

Using the explicit solution formula~\eqref{E:w} for $w$ we obtain the iteration
\begin{equation}
\label{E:defT}
	Q^{(j)} = T(Q^{(j-1)}) 
	:= Q^{(j-1)} \iint_\Omega k(x,s) e^{-Q^{(j-1)}\Lambda(x,s)}\, d^2(s,x).
\end{equation}
Here $T:[0,1]\to [0,1]$ is the iteration map related to $F$ by $T(Q)=Q\cdot (1-F(Q))$. It is obvious, that $T(0)=0$ and $T(Q^\ast)=Q^\ast$ provided $Q^\ast$ exists as well as $T(1)=1-F(1)<1$ due to Lemma~\ref{L:F0:F1}.

\begin{lemma}
\label{L:dT}
The iteration map $T$ is differentiable and
\begin{enumerate}
\item $T'(0) = \calR_0$.
\item $T'(Q^\ast) < 1$.
\end{enumerate}
\end{lemma}
\begin{proof}
The first statement is obvious using $T'(Q)=1-F(Q)-Q\cdot F'(Q)$ and the second one follows from $Q^\ast>0$ and $F'(Q)>0$ for $Q>0$.
\end{proof}

\begin{lemma}
\label{L:dTstrict}
For $\calR_0<2e\simeq 5.47$, we have $-1< T'(Q^\ast)<1$.
\end{lemma}
\begin{proof}
Let $\calR_0<2e$. Using $F(Q^\ast)=0$ and Lemma~\ref{L:dF} we get
\begin{equation*}
	-T'(Q^\ast) = F(Q^\ast)+Q^\ast\cdot F'(Q^\ast)-1 
			\le Q^\ast\cdot \frac{\calR_0\, e^{-1}}{Q^\ast} - 1 
			= \calR_0\, e^{-1} -1 < 1.
\end{equation*}
\end{proof}

Summarizing these findings, we arrive at the following result.

\begin{theorem} The fixed point iteration~\eqref{E:defT} has the following convergence properties:
\begin{enumerate}
\item For $0<\calR_0<1$, the trivial fixed point $Q=0$ is the only fixed point in $[0,1]$ and the fixed point iteration is locally convergent to $Q=0$, i.e.~$T\lb Q^{(j)}\rb \to 0$.
\item For $1<\calR_0<2e$, a non-trivial fixed point $Q^\ast\in (0,1)$ exists and the fixed point iteration is locally convergent to $Q^\ast$. The trivial fixed point $Q=0$ is locally repelling in this case.
\end{enumerate}
\end{theorem}
\begin{proof}
Let $\bar{Q}$ denote any of the fixed points and define 
\begin{align*}
	\Delta^{(j)} &:= Q^{(j)} - \bar{Q} 
					= T\lb \bar{Q}+\Delta^{(j-1)}\rb - T(\bar{Q})
\intertext{A linearization yields}
	\Delta^{(j)} &= \Delta^{(j-1)}\cdot T'(\bar{Q}) + \calO\left((\Delta^{(j-1)})^2\right).
\end{align*}
For $0<\calR_0<1$ we only acquire the fixed point $\bar{Q}=0$ and $T'(0)=\calR_0<1$; hence it is locally attractive. If $\calR_0>1$, $\bar{Q}=0$ is still a fixed point, but it turns to be locally repelling due to $T'(0)=\calR_0>1$. For $\calR_0>1$, the second fixed point $\bar{Q}=Q^\ast$ appears and for $\calR_0<2e$ we have
 $-1<T'(Q^\ast)<1$, i.e.~local attractivity. 
 \end{proof}

In the case of $\calR_0\ge 2e$ we fail to show the local convergence of the fixed point iteration. However, as seen later, the practical relevant cases solely cover the basic reproductive number that is far below this technical bound of $2e$. In Figure~\ref{F:ContR} we show a contour plot for $\calR_0$ depending on the rate $\gamma$ and the reference value for maximal transmission $x_p$. The shape parameters $\theta_0$, $\theta_1$ and $\alpha$ of the transmission function~\eqref{E:thetaGauss} agree with the optimized values given in Table~\ref{T:ParamVals}.

\begin{figure}[htb]
\centerline{\includegraphics[width=0.67\textwidth]{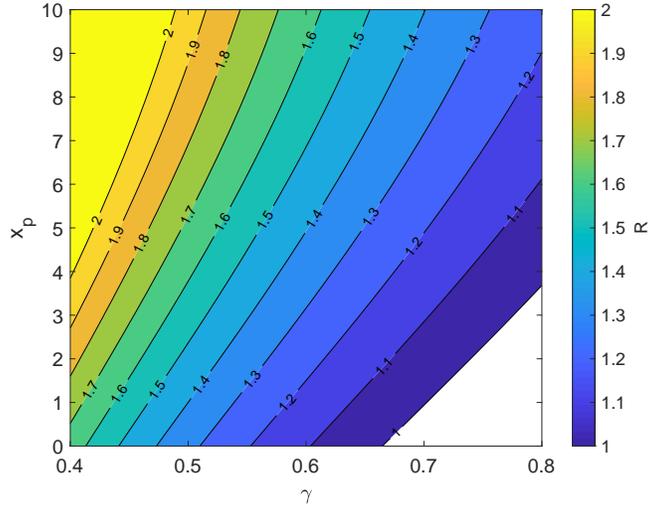}}
\caption{\label{F:ContR} Contour plot for $\calR_0$ depending on $\gamma$ and $x_p$. For $\calR_0>1$, the fixed point problem~\eqref{E:equiSI-PDE} admits a non--trivial solution.}
\end{figure}

\subsection{Approximate analytical solution}

Here, we consider an approximate solution for a simplified version of the fixed point problem~\eqref{E:equiSI-PDE}. We assume a transmission rate $\theta(x)\equiv \theta_0$ independent of the age $x$ and a linear behaviour of the population $p(x)=N_0\lb 1-x/A\rb$. In this case, $N=\int_0^{A} p(x)\, dx = \frac{N_0 A}{2}$ and the basic reproductive number $\calR_0$ equals
\begin{equation}
\label{E:R0:simple}
	\calR_0 = \frac{\theta_0 N_0}{N} \int_0^A e^{-\gamma x} \lb 1-\frac{x}{A}\rb \int_0^x e^{\gamma s}\, ds \, dx = \frac{k}{\delta^3} \lsb (\delta-1)^2+1-2e^{-\delta}\rsb =: k\cdot f(\delta) 
\end{equation}
for $k=\theta_0 A$ and $\delta=A\gamma$ and $f(\delta):= 2 \lb 1-\delta+\frac{\delta^2}{2}-e^{-\delta}\rb / \delta^3$.

Introducing 
\begin{align*}
	\alpha &:= \frac{\theta_0}{N} \int_0^{A} p(x)\, w(x)\, dx 
		= \frac{2\theta_0}{A N_0} \int_0^{A} p(x)\, w(x)\, dx,
\intertext{which is actually a constant, the ODE~\eqref{E:ODE:w} simplifies to}
	w' &= -(\gamma+\alpha) w + \alpha, \quad w(0)=0
\intertext{with the solution}
	w(x) &= \frac{\alpha}{\gamma+\alpha} \lb 1-e^{-(\gamma+\alpha)x}\rb.
\end{align*}
Hence, we are able to compute $\alpha$ and obtain
\begin{align*}
	\alpha &= \frac{2 \theta_0}{A} \int_0^{A}
	 \lb 1-\frac{x}{A}\rb \cdot \frac{\alpha}{\gamma+\alpha}
	  \cdot \lb 1-e^{-(\gamma+\alpha)x} \rb \; dx \\
	 &= \alpha \frac{\theta_0}{A^2 (\gamma+\alpha)^3}
		\lsb \lb A(\gamma+\alpha)-1\rb^2+1 -2e^{-A(\gamma+\alpha)}\rsb.
\end{align*}
Obviously $\alpha=0$, i.e.~$w(x)=0$ for all $x$, is a solution corresponding to the trivial disease free equilibrium. Other biologically reasonable solutions require $w>0$ and hence $\alpha>0$. To determine these solutions, we introduce $z:=A(\gamma+\alpha)$ and arrive at the equation
\begin{equation}
\label{E:zfull}
	z^3 = k\lsb (z-1)^2+1 -2e^{-z}\rsb.
\end{equation}
Here, biologically reasonable solutions with $\alpha>0$  require that $z>A \gamma=k \cdot \gamma/\theta_0$.

Typically, $0.1\lesssim \theta_0 \lesssim 1$ and  $A=100$, hence $k\gg 1$. To seek the roots of~\eqref{E:zfull} depending on $k$, we use the scaling $z=k^b y$. Balancing the individual terms in~\eqref{E:zfull} yields reasonable results only for $b=1$, i.e.~$z=ky$. Introducing $\eps:=1/k \ll 1$ and canceling the exponentially small term $2\eps^3 e^{-y/\eps}$, we arrive at the cubic equation
\begin{equation}
\label{E:cubicy}
	y^3	= y^2 -2\eps y + 2\eps^2. 
\end{equation}
Using a regular asymptotic expansion $y \simeq y_0 + \eps y_1 + \eps^2y_2$, we obtain
\begin{subequations}
\begin{alignat}{5}
\calO(\eps^0): &\qquad& y_0^3 &= y_0^2 &\quad\leadsto\quad 
		& y_0=0 \text{ or } y_0=1.
\intertext{Here, $y_0=1$ is the only interesting solution not leading to the disease-free equilibrium. We thus proceed further with this solution. For the next orders we obtain}
\calO(\eps^1): && 3 y_0^2 y_1 &= 2y_0y_1 -2y_0 
	&\quad\leadsto\quad& y_1 = -2, \\
\calO(\eps^2): && 3 y_0^2 y_2 + 3 y_0 y_1^2 &= y_1^2-2y_0y_2-2y_1+2
	&\quad\leadsto\quad& y_2=-2.
\end{alignat}
\end{subequations}
Therefore, we obtain the expansion
\begin{equation}
	\label{E:asyz}
	z \simeq k -2 -\frac{2}{k} + \calO(1/k^2) 
		\simeq k -2 + \calO(1/k).
\end{equation}
In Figure~\ref{F:Sol_asy_z}, we show a graphical comparison of the solution for the full equation~\eqref{E:zfull}, the cubic equation~\eqref{E:cubicy} (written for $z$ instead of $y$), and the asymptotic expansion~\eqref{E:asyz}.

\begin{figure}[htb]
\centerline{\includegraphics[width=0.67\textwidth]{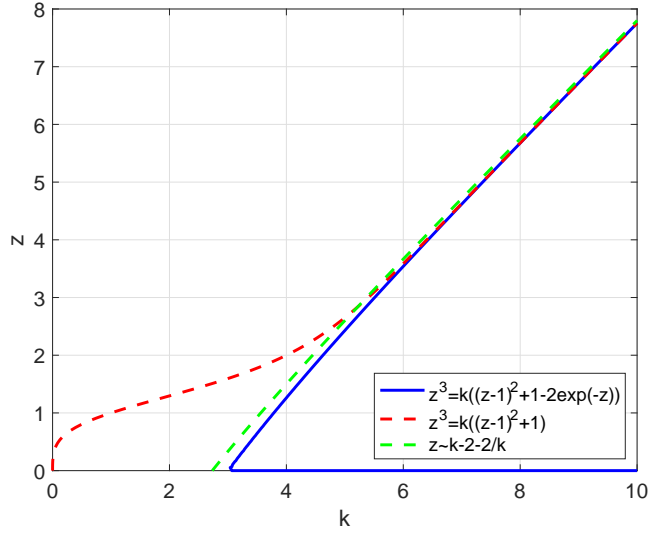}}
\caption{\label{F:Sol_asy_z} Solutions $z$ vs.~$k$ of the full equation~\eqref{E:zfull} (solid blue), the cubic equation~\eqref{E:cubicy} written for $z$ instead of $y$ (dashed red) and the asymptotic expansion~\eqref{E:asyz} (dashed green). For $k\gg 1$, the asymptotic expansion $z\simeq k-2-2/k$ provides a good approximation of the non-trivial solution of the transcendental equation $z^3 = k\lsb (z-1)^2+1 -2e^{-z}\rsb$.}
\end{figure}

Combining the expansion~\eqref{E:asyz} with the requirement $z>k\cdot \gamma/\theta_0$, we observe that for $\gamma>\theta_0$, there cannot be any biologically relevant solution of the fixed point equation besides the trivial disease-free equilibrium. This matches with an analysis based on the basic reproductive ratio, see~\eqref{E:R0:simple}. For $\delta=A\gamma>A\theta_0=k$, we get that $\calR_0 = k\cdot f(\delta) < k \cdot f(k) < 1$ and hence only the trivial disease-free equilibrium exists.

Using the two-term expansion $z\simeq k-2$ we are able to determine
\begin{align}
	\alpha & \simeq \theta_0-\gamma -\frac{2}{A}  
\intertext{With this approximate value of $\alpha$, we get}
	w(x) &= \frac{\alpha}{\gamma+\alpha}\lb 1-e^{-(\gamma+\alpha)x} \rb  = \lsb 1- \frac{A \gamma}{2-A\theta_0} \rsb \lb 1-e^{-\theta_0 x} e^{2x/A}\rb
\intertext{and}
	Q &= \frac{\alpha}{\theta_0} 
	= 1-\frac{\gamma}{\theta_0} -\frac{2}{A \theta_0}.
\end{align}
Note, that the endemic equilibrium $1-\frac{\gamma}{\theta_0}$ of a standard SI model without age structure matches with the expansion of $Q$ for $k=A\theta_0\gg 1$.

\section{Simulations}

\subsection{Data of hospitalized cases}
The data accompanying subsequent simulations and parameter estimation consists of the hospitalized dengue cases from 2009 to 2014 from the city of Semarang, Indonesia. Figure~\ref{F:Data} provides an overview of the relative number of registered cases for 14 age classes. During the observation period, the annual total number of dengue cases varied significantly. The peak year 2010 recorded a total of more than 5.500 cases compared to the minimum of just 1.250 cases in 2012. However, the data clearly show, that the \emph{relative} occurrence of dengue cases within each age class remained rather similar for all six years despite the large variations in the total number of cases within each year. Therefore, we henceforth neglect the variations from one year to another and instead consider the average over the given six year period in final comparison of SI and SIR model. 

\begin{figure}[htb]
\centerline{\includegraphics[width=.67\textwidth]{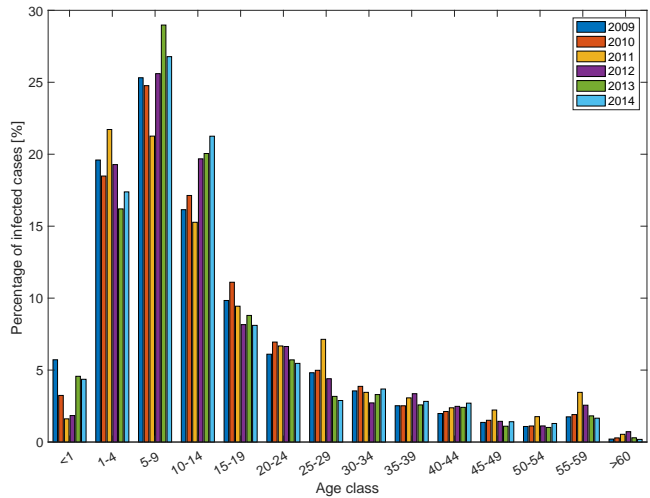}}
\caption{\label{F:Data} Hospitalized dengue cases in Semarang from 2009 to 2014. Shown is the percentage of cases (relative to the total number of infections) within each of the $14$ age classes.}
\end{figure}

\subsection{Parameter estimation}

In a consecutive step, we will use the age-structured SI model~\eqref{E:equiSI-PDE} to predict the age distribution of dengue cases. For that purpose we need to identify the parameters $c:=(\theta_0, \theta_1, x_p, \sigma)$ shaping the transmission rate $\theta$, as well as the rate of loss of immunity $\gamma>0$. We estimate these model parameters, such that the model solution $w$ and the observed data agree in the least-squared sense:
\begin{align*}
	& \min_{(c,\gamma)\in \R^4\times \R_+} \norm{w-data}^2 \\
	\text{s.t.} \quad
	w' & = -\gamma w + \theta(x;c)\cdot (1-w)\cdot Q/N,\ w(0)=0.
\end{align*}
Under certain discretization strategies for the objective function and constraint, this infinite-dimensional transforms into a finite-dimensional minimization problem that can be solved by any standard method. The numerical results of the minimization, using either the average of dengue cases in the six years 2009--2014 or just one single year are given in Table~\ref{T:ParamVals}.

\begin{table}[htb]
\begin{tabular}{r|llllll}
Parameter 	& $\theta_0$	& $\theta_1$	& $x_p$		& $\alpha$ 	& $\gamma$	& $\norm{u-data}$ \\
\hline
Opt.~2009--2014	& $0.1821$	& $2.1135$	& $-0.5209$	& $10.3164$	& $0.6412$	& $135.55$ \\ \hline
Optimized 2009 & $0.2226$	& $1.2351$	& $2.6900$	& $6.2967$	& $0.4422$	& $369.7005$ \\
Optimized 2010	& $0.1883$	& $2.6950$	& $-4.7281$	& $13.7437$	& $0.7810$	& $235.5216$ \\
Optimized 2011	& $0.1777$	& $2.1033$	& $-11.4227$	& $19.5034$	& $0.6753$	& $119.1197$ \\
Optimized 2012	& $0.2083$	& $1.3271$	& $1.7922$	& $6.7265$	& $0.4395$	& $81.4142$ \\
Optimized 2013	& $0.0919$	& $1.5583$	& $4.0973$	& $6.8862$	& $0.4325$	& $155.7308$ \\
Optimized 2014	& $0.1112$	& $1.4544$	& $5.0660$	& $7.9310$	& $0.4849$	& $145.0712$ \\
\end{tabular}
\medskip
\caption{\label{T:ParamVals} Optimized parameter values for the transmission rate $\theta$ and rate of loss of immunity $\gamma$. As reference data we used either the average over all six years or just the data from each single year.}
\end{table}

The optimized parameter given in Table~\ref{T:ParamVals} show that the peak in the transmission rate actually occurs at young ages. The parameter $x_p$ being slightly negative implies that the maximum of $\theta$ occurs for newborns and with increasing age, the risk of getting infected is decreasing. Based on this model for the transmission rate, we are able to reproduce the observed data quite well. Using the annual averaged data, the maximal absolute difference between the simulation and the data equals to $33$ cases and occurs for the age cohort between $15$ and $20$ years.  

In Figure~\ref{F:infect_sim}, we depict the simulation results for the different age cohorts. The given data include infected cases of age less than $1$ year. This information is excluded in our model since we have assumed no vertical transmission. In the senior age classes, the given data only reports for age older than 60 years with no further subdivision.

\begin{figure}[htb]
\centerline{\includegraphics[width=0.67\textwidth]{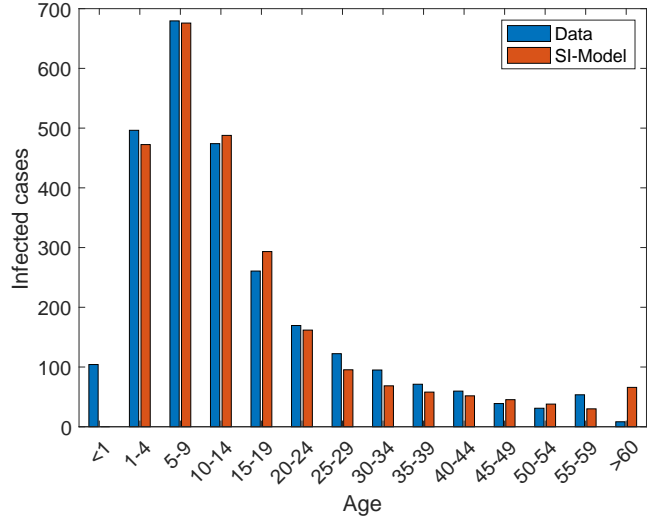}}
\caption{\label{F:infect_sim} Comparison of data (average of observed cases in the years $2009$--$2014$) and simulation for stationary state of infected class.}
\end{figure}

\subsection{SIR model}

In a last step, we extend the previous SI model~\eqref{E:equiSI-PDE} to an SIR model
\begin{subequations}
\label{E:equiSIR-PDE}
\begin{alignat*}{2}
	w' &= - \gamma w + \theta(x;c)\cdot (1-w-z)\cdot \frac{Q}{N}, \qquad & w(0) &=0, \\
	z' &= \gamma w\;, & z(0)&=0, \\
	Q &= \int_0^{A} p(y) w(y)\, dy.
\end{alignat*}
\end{subequations}
This fixed point problem is again solved by an iterative method analogue to the SI model and again the model parameters $(c,\gamma)$ are identified based on the given data averaged over the six years period 2009--2014. 

In Figure~\ref{F:SI_SIR} we show a comparison of averaged field data with the simulation results for the optimized SI and SIR model. Table~\ref{T:parSI_SIR} provides the according optimal parameter values for the SI and SIR model. Note, that the SIR model produces a slightly better agreement with the averaged data reducing the $L_2$--error by around $5\%$.

\begin{figure}[htb]
\centerline{\includegraphics[width=0.67\textwidth]{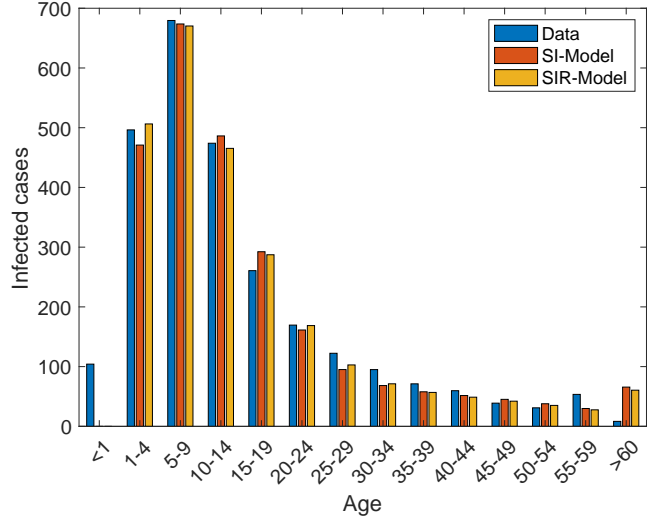}}
\caption{\label{F:SI_SIR} Comparison average data and simulations using the SI and SIR model.}
\end{figure}

\begin{table}[htb]
\begin{tabular}{r|llllll}
Parameter 	& $\theta_0$	& $\theta_1$	& $x_p$		& $\alpha$ 	& $\gamma$	& $\norm{u-data}$ \\
\hline
Optimal SI	& $0.1821$	& $2.1135$	& $-0.5209$	& $10.3164$	& $0.6412$	& $135.55$ \\
Optimal SIR	& $0.1807$	& $2.6877$	& $-5.8795$	& $12.7946$	& $0.6511$	& $128.4848$
\end{tabular}
\medskip
\caption{\label{T:parSI_SIR} Optimized parameter values for $\theta$ and $\gamma$ in the SI and SIR model.}
\end{table}

For the sake of comparison, we consider a standard \emph{age-independent} SIR model
\begin{align*}
	I' &= \frac{\hat{\theta}}{N}(N-I-R)I- (\hat{\gamma}+\mu) I \\
	R' &= \hat{\gamma} I - \mu R
\intertext{with equilibrium $(I^\ast, R^\ast)$ where the infective state is given by}
	I^\ast &= \hat{\mu} N \lb \frac{1}{\hat{\mu}+\gamma} - \frac{1}{\hat{\theta}}\rb\;.
\end{align*}
The mortality rate and transmission rate in the age-independent model are called $\hat{\mu}$ and $\hat{\theta}$. The rate of loss of immunity $\gamma$ is assumed to be equal in both the age-structured and age-independent model. To compare the results given in Table~\eqref{T:parSI_SIR} for the equilibrium distribution of the age--structured SIR--Model~\eqref{E:equiSIR-PDE} with an \emph{age--independent} model, we balance the total number of newborns $p(0)$ in the age-structured model and $\hat{\mu} N$ in the age-independent model to obtain $\hat{\mu}=\frac{p(0)}{N}\simeq 0.018$.

Using the parameters for the age distribution~\eqref{E:pinfty} and the simulation result $Q=2537$, we set $I^\ast=Q$ and get
\begin{equation}
	\hat{\theta} = \lb \frac{1}{\hat{\mu}+\gamma} - \frac{Q}{\hat{\mu} N}\rb^{-1} 
		=  \lb \frac{N}{p(0)+\gamma N} - \frac{Q}{p(0)}\rb^{-1} \simeq 0.711\;.
\end{equation}
Whether this age-independent transmission rate $\hat{\theta}$ can be directly obtained by analytic computations just using the PDE model is left for future research.

\section{Conclusion}

Medical statistics have shown a significant age-dependence in dengue infections. Extending the standard ODE-governed SI(R) models to PDE-governed models incorporates the age structure into the mathematical equations. Equilibrium distributions are then governed by the fixed points of resulting integro-differential equations. An extension of the concept of the basic reproductive number allows to characterize the parameter regimes, where just disease-free or also endemic equilibria exist. Furthermore, the convergence of an iteration scheme to compute the endemic equilibrium has been shown. Existing data from the city of Semarang, Indonesia have been used to validate the steady-state distribution and identify unobservable model parameters. With respect to the $L_2$--error, the output of the age-structured SI model and given data show a high level of agreement. Extending the SI model to a more realistic SIR model yields a further reduction of the $L_2$--error. Comparing the age-structured with the age-independent model allows to determine equivalent effective age-independent transmission rates. Whether and how those rates can be directly obtained from analytic computations is left for future research.

\section*{Acknowledgment}
The authors would like to thank their colleague, Sutimin, from Diponegoro University, Semarang from providing the field data. The first author would like to thank National Science Foundation -- Sri Lanka for funding his research visit to University of Koblenz-Landau, Germany in 2019.

\bibliographystyle{elsarticle-num}
\bibliography{agemodel-GGW}

\end{document}